\documentclass[12pt]{amsart}
\usepackage{fullpage} 
\usepackage[T1]{fontenc}
\usepackage[utf8]{inputenc}

\usepackage{amsfonts,amsmath,amsthm,amssymb}
\usepackage{enumitem} 
\usepackage{stix} %apparently it is obsolet
\usepackage[backend=biber,style=numeric,hyperref=true,isbn=false,doi=false]{biblatex}
\newtheorem{theorem}{Theorem}
\newtheorem{lemma}[theorem]{Lemma}
\newtheorem{proposition}[theorem]{Proposition}
\newtheorem{fact}[theorem]{Fact}
\newtheorem{problem}[theorem]{Problem}

\theoremstyle{remark}
\newtheorem{remark}[theorem]{Remark}
\newtheorem{example}[theorem]{Example}
\newtheorem{definition}[theorem]{Definition}

\usepackage[norelsize,boxed,vlined]{algorithm2e}

\usepackage{tikz}
\usetikzlibrary{automata, arrows.meta, positioning}
\usetikzlibrary{shapes.multipart}
\usetikzlibrary{calc}

\title{Kunz languages for numerical semigroups are context sensitive}
\author{Manuel Delgado
\email{mdelgado@fc.up.pt}
\and
Jaume Usó i Cubertorer
\email{jaumeusoicu@gmail.com}
}
\address{CMUP--Centro de Matemática da Universidade do Porto\\
Departamento de Matemática, Faculdade de Ciências\\
Universidade do Porto\\
Rua do Campo Alegre s/n, 4169– 007 Porto, Portugal}
\thanks{The first author was partially supported by CMUP, a member of LASI, which is financed by national funds through FCT – Fundação
  para a Ciência e a Tecnologia, I.P., under the projects with reference UIDB/00144/2020 and UIDP/00144/2020. He also acknowledges the Proyecto de Excelencia de la Junta de Andalucía (ProyExcel 00868).}
\date{\today}

\begin{document}

\begin{abstract}
   There is a one-to-one and onto correspondence between the class of numerical semigroups of depth \(n\), where \(n\) is an integer, and a certain language over the alphabet \(\{1,\ldots,n\}\) which we call a Kunz language of depth \(n\). The Kunz language associated with the numerical semigroups of depth \(2\) is the regular language \(\{1,2\}^*2\{1,2\}^*\). We prove that Kunz languages associated with numerical semigroups of larger depth are context-sensitive but not regular.
\end{abstract}

\maketitle

\section{Introduction}\label{sec:introduction}
This paper exhibits an application of the area of formal languages in the area of numerical semigroups.

We assume that most experts in one of the referred areas are not experts in the other. This leads us to try to make the paper easily readable by non-experts.

Let \(\mathbb{N}=\mathbb{Z}_{\ge 0}=\{0,1,2,\ldots\}\) be the set of nonnegative integers.
A \emph{numerical semigroup} \(S\) is a subset of \(\mathbb{N}\) such that \(0\in S\), \(S\) is closed under addition and  \(\mathbb{N}\setminus S\) is finite.

Our motivation is to formalize in some way why certain problems in numerical semigroups may have different degrees of difficulty depending on the families of numerical semigroups under consideration.

Although not discussed in this article, one example of such a problem is the classical Frobenius problem: given relatively prime positive integers \(a_1,a_2,\ldots,a_n\) (which can be thought of as generators of a numerical semigroup), ﬁnd the largest integer, denoted \(F(a_1,a_2,\ldots,a_n)\), that is not representable as a nonnegative integer combination of \(a_1,a_2,\ldots,a_n\). This problem is easy to solve when \(n=2\); in this case there is a formula: \(F(a_1,a_2)=a_1a_2-a_1-a_2\). When at least \(3\) integers are needed to generate the numerical semigroup, there is no `closed' formula for the Frobenius number. Ramirez-Alfonsin's book~\cite{Ramirez-Alfonsin2005Book-Diophantine} contains these and numerous other results.
Problems in numerical semigroups tightly related to the Frobenius problem tend to be easy for numerical semigroups generated by two coprime integers and hard for those requiring more generators.
Also, various partial results for numerical semigroups generated by three integers have no correspondents known for numerical semigroups requiring more generators.

To a numerical semigroup, one can associate a word; thus, one can associate a formal language with a class of numerical semigroups. In this article, we will refer to certain classes of numerical semigroups and analyze their associated languages. Some of these are regular; some others are not even context-free.

Regular languages are objectively very simple, since they are accepted by ﬁnite state automata (the simplest abstract computer model, not even having memory) and a non-regular language is not.
We show that the languages considered are all context-sensitive (accepted by linear bounded automata).

Benefiting from the fact that we do not need to enter deep into any of the above-mentioned areas, we made the paper self-contained through a section that includes definitions, motivation and preliminaries. There we devote a subsection to recalling the basics of numerical semigroups. Some aspects of language theory are recalled next. Representing numerical semigroups by words is the subject of the third subsection of preliminaries. There we define Kunz languages as sets of words over \(\mathbb{N}\) satisfying some conditions.

The final section contains our results. 
We discuss where Kunz languages fit into the Chomsky hierarchy. 

\section{Preliminaries, definitions and motivation}\label{sec:preliminaries}
In this section we introduce the concepts, terminology, and notation to be used in the third section. At the same time, we explain the motivation for the research presented. Also, we include examples to illustrate the concepts introduced and prepare material to be used in the last section.

We divide the section into three subsections. The first contains the material on numerical semigroups needed. The second recalls, somewhat informally, the basics on formal languages. In the third subsection we define Kunz words, a concept already appearing in the literature, later used to define Kunz languages.
%%%
\subsection{Numerical semigroups}\label{subsec:numericalsgps}
  Recall that a numerical semigroup \(S\) is a subset of \(\mathbb{N}\) containing \(0\), closed under addition and such that \(\mathbb{N}\setminus S\) is finite.
 The number of elements of \(\mathbb{N}\setminus S\) is called the \emph{genus} of \(S\) and is denoted \(\operatorname{g}(S)\).
As a general reference on numerical semigroups, we refer to~\cite{RosalesGarcia2009Book-Numerical}. All results referred to as well-known are taken from there.

The study of numerical semigroups frequently involves investigating some notable elements and relations between them. Among the notable elements appearing in this text, there is the least positive element of the semigroup and the smallest element of the semigroup such that all larger integers are likewise elements of the semigroup.
These are called, respectively, the \emph{multiplicity} and the \emph{conductor} of the semigroup. For the matter of notation, for any given numerical semigroup \(S\), we shall use \(m=\operatorname{m}(S)=\min(S\setminus\{0\})\)  and \(c=\operatorname{c}(S)=\min \{x\in S\mid x+\mathbb{N}\subseteq S\}\), respectively, for the multiplicity and the conductor of \(S\). The argument \(S\) is usually understood and omitted when readability is unaffected.

For \(x\in\mathbb{R}\), we denote the \emph{ceiling} of \(x\), \(\min\{r\in \mathbb{Z}\mid x\le r\}\), by \(\lceil x \rceil\), as is usual.
Although not belonging to the numerical semigroup being studied, the \emph{Frobenius number} of \(S\), defined by \(\operatorname{F}(S)=\operatorname{c}(S)-1\), and the \emph{depth} of \(S\), defined as
  \(
  \left\lceil\frac{\operatorname{c}(S)}{\operatorname{m}(S)}\right\rceil
\), are also of great importance.
As \(\operatorname{c}(\mathbb{N})=0\) and  \(\operatorname{m}(\mathbb{N})=1>0\), the depth of \(\mathbb{N}\) is \(0\); all other numerical semigroups have positive depth.

\medskip

Let \(S\) be a numerical semigroup with multiplicity \(m\).
The \emph{Apéry set} of \(S\), with respect to its multiplicity, is
\[
\operatorname{Ap}(S)=\{s \in S \mid s-m \notin S\}.
\]
It is well known and easy to see that \(\operatorname{Ap}(S)\) has \(m\) elements, one per residue class modulo~\(m\). To be more precise, we can write \(\operatorname{Ap}(S)=\left\{0, a_1, \ldots, a_{m-1}\right\}\), where \(a_i=\min \{x \in S \mid x \equiv i \bmod m\}\), so that each \(a_i=k_i m+i\) for some positive integer \(k_i\) (\(0<i<m\)). Each of these \(k_i\) is called a \emph{Kunz coordinate of \(S\)} (see below for the reason of the terminology).
\medskip

It is also well-known that \(\operatorname{Ap}(S)\cup\{m\}\) generates \(S\). That is,
\[S=\{n_0m+n_1a_1+\cdots + n_{m-1}a_{m-1}\mid n_0,n_1, \ldots,n_{m-1}\in \mathbb{N}\}.\]
This shows, in particular, that numerical semigroups are finitely generated.

Sometimes the notation \(S= \langle m, a_1, \ldots, a_{m-1}\rangle\) is used, even when \(\operatorname{Ap}(S)\cup\{m\}\) is not a minimal set of generators.
It is well-known that a numerical semigroup \(S\) has a unique minimal set of generators. Its cardinality, called the \emph{embedding dimension} of \(S\), is denoted \(\operatorname{e}(S)\).

\medskip

Let \(S\) be a numerical semigroup with multiplicity \(m\). We write \(\mathcal{K}(S)=\left(k_1, k_2, \ldots, k_{m-1}\right)\) and use the convention \(\mathcal{K}(\mathbb{N}) = ()\). The (possibly empty) tuple \(\mathcal{K}(S)\)
is called the \emph{Kunz coordinates tuple} of \(S\). 

The sum of the Apéry elements \(a_i=k_i m+i\) and \(a_j=k_j m+j\), belongs to the congruence class of \(i+j\) modulo \(m\); in fact, one can write: \(a_i + a_j=(k_i+k_j)m+i +j=(k_i+k_j+1)m+i +j-m\). By filling in a few more details, it can be shown (see \cite{RosalesGarcia-SanchezGarcia-GarciaBranco2002JLMS2-Systems}) that an integer tuple \(\left(x_1, \ldots, x_{m-1}\right)\) is the Kunz coordinates tuple of a numerical semigroup with multiplicity \(m\) if and only if
\begin{align}\label{eq:kunz-inequalities}
	x_i & \geq 1 & \text{ for } 1 & \leq i \leq m-1, \nonumber\\
	x_i+x_j & \geq x_{i+j} & \text{ for } 1 & \leq i \leq j \leq m-1 \text{ with } i+j<m, \nonumber\\
	%x_i+x_j+1 & \geq x_{i+j -1}& \text{ for } 1 & \leq i \leq j \leq m-1 \text{ with } i+j>m
    x_i+x_j+1 & \geq x_{i+j-m}& \text{ for } 1 & \leq i \leq j \leq m-1 \text{ with } i+j>m.
\end{align}

\medskip

The map \(\mathcal{K}\) is a bijection between numerical semigroups 
and Kunz coordinate tuples. Due to this bijection, which provides a (natural) representation of a numerical semigroup, we will sometimes identify a numerical semigroup with its Kunz coordinates tuple.
\medskip

Kunz coordinates play a fundamental role in various aspects of the theory of numerical semigroups. As an example, we point out that for any fixed $m$ the inequalities in (\ref{eq:kunz-inequalities}) define a polyhedron (which was introduced independently in \cite{Kunz1987Book-Uber} and \cite{RosalesGarcia-SanchezGarcia-GarciaBranco2002JLMS2-Systems}). Its points with integer coordinates are in a natural bijection (via the map \(\mathcal{K}\) defined above) with numerical semigroups with multiplicity~\(m\). This fact was explored in~\cite{BrunsGarcia-SanchezONeillWilburne2020IJAC-Wilfs} (respectively, \cite{KliemStump2022DCG-new} ) to show that there is no counterexample to a long-standing conjecture among numerical semigroups with multiplicity up to~\(18\) (respectively,~\(19\)). This conjecture, known as Wilf's conjecture, was raised as a question in \cite{Wilf1978AMM-circle}.
The question asks whether, for every numerical semigroup \(S\), \(\frac{\operatorname{g}(S)}{\operatorname{c}(S)}\) is at most \(1-\frac{1}{\operatorname{e}(S)}\).
  It appears in the literature in various equivalent forms and is nowadays an important topic of research in the area of numerical semigroups; see \cite{Delgado2020-survey} for a survey.

\bigskip

Let \(q = \operatorname{q}(S)\) be the maximum of the Kunz coordinates of \(S\). By definition, there exists an integer~\(r\) such that \(0<r<m\) and \(\max\operatorname{Ap}(S) = qm +r\). It is well known and easy to observe that \(\max\operatorname{Ap}(S) = \operatorname{F}(S)+m = \operatorname{c}(S)+m-1\). Thus \(qm+r=c+m-1\), which implies that \(q=\left\lceil\frac{c}{m}\right\rceil\). Therefore the depth of \(S\) is its maximum Kunz coordinate. 
\medskip

For a positive integer \(q\), denote by \(\mathcal{S}_q\) the class of numerical semigroups of depth~\(q\). 

%%%
\subsection{Formal languages}\label{subsec:formal-languages}
Here we quickly overview formal languages, grammars, and machines. The reader can find further details in any introductory book on the subject. For example, the book by Linz~\cite{Linz2011Book-Introduction} serves this purpose.
\medskip

Given a nonempty set of symbols \(\Sigma\), we denote by \(\Sigma^*\) the set consisting of all possible results of concatenating a finite number of symbols of \(\Sigma\), along with the empty word.
Then, a language \(L\) over an alphabet \(\Sigma\) is a subset of \(\Sigma^*\), whose elements are called \emph{words}. The number of symbols in a word \(w\) is called the \emph{length} of the word and is denoted \(|w|\). Sometimes, we will refer to all possible nonempty words over an alphabet, for which the notation \(\Sigma^+\) will be used.
\medskip

There are two essential ways to describe a language: giving a set of rules that explain how to build the strings in that language (grammars), or recognizing which words belong to the language and which do not (automata).
\medskip

A grammar \(G\) is defined as a quadruple \(G = (V,T,S,P)\) where \(V\) is a finite set of variables, \(T\) is a finite set of terminal symbols such that \(V\) and \(T\) are nonempty and disjoint; \(S\in V\) is the start variable and  \(P\) is a finite set of productions of the form \(x \to y\) (read `\(x\) produces \(y\)'),  where \(x\in (V \cup T)^+\) contains at least one variable and \(y \in (V \cup T)^*\).
\medskip

During the 1950s, Noam Chomsky introduced a hierarchy of four types of grammars, each with more restricted production rules than the previous one. This hierarchy is known as the \emph{Chomsky} or \emph{Chomsky-Schützenberger} hierarchy. As each grammar produces a language, this hierarchy of grammars translates to a hierarchy in languages, which we will use to classify numerical semigroups. 
These four classes of languages in Chomsky's hierarchy are, in increasing order of complexity: regular, context-free, context-sensitive and recursively enumerable.

\begin{definition}\label{def:ex-grammars}
    Let \(G = (V,T,S,P)\) be a grammar.
    \begin{enumerate}
        \item The grammar \(G\) is said to be \emph{right linear} if all its productions are of the form 
        \[A\to xB \; | \; x\]
        with \(A,B\in V\) and \(x\in T^*\);
        \item The grammar \(G\) is said to be \emph{context-free} if all its productions are of the form 
        \[A \to x\]
        with \(A\in V\), \(x\in (V\cup T)^*\). 
        \item The grammar \(G\) is said to be \emph{context-sensitive} if all its productions are of the form 
        \[xAy \to xvy\]
        with \(A\in V\), \(x,y\in ((V\cup T)\setminus \{S\})^*\) and \(v\in ((V\cup T)\setminus \{S\})^+\). 
    \end{enumerate}
\end{definition}
The name context-sensitive is justified by the fact that \(x\) and \(y\) form the context of \(A\): they determine whether \(A\) can be replaced with \(v\). Notice the contrast with context-free grammars, where no context is present: the left-hand side of a production is just a non-terminal.

If we have a production \(x \to y\) and a word \(w = uxv\), our production is applicable to \(w\) and lets us replace \(x\) by \(y\) in \(w\) obtaining the new word \(z = uyv\); this is denoted by \(w \Rightarrow z\).
A production can be used whenever it is applicable, and it can be applied as often as desired. If
\[w_1 \Rightarrow w_2 \Rightarrow \cdots \Rightarrow w_n\]
we write \[w_1 \overset{*}{\Rightarrow} w_n,\] the \(*\) indicating that an unspecified number of steps (including \(0\)) can be taken. 
The language defined by a grammar \(G\) is   \[L(G) = \{ w\in T^* | S \overset{*}{\Rightarrow} w \}.\] 

Right-linear grammars define \emph{regular languages}; \emph{context-free languages} are defined by context-free grammars and \emph{context-sensitive languages} are defined by context-sensitive grammars.
 
 The same classes of languages mentioned can be obtained using automata; each of these families of languages is recognized by some type of automata. The approach using automata is the one used in this paper, but in some cases using grammars is more natural.

 At the bottom of the Chomsky hierarchy, we have regular languages. These are languages recognized by deterministic finite accepters, or dfa, which are the simplest machines considered here.
 Deterministic finite accepters simply consist of a finite number of states and instructions about how to move between them. These automata have no memory. Frequently, dfa's are thought of as transition graphs.
 \begin{example}\label{ex:dfaK2}
    The transition graph depicted in Figure~\ref{fig:dfaK2} represents a dfa with two states (\(s_0\) is the initial state, and \(s_1\) is the final state) and its alphabet is \(\{1,2\}\). The words that can be read from the initial state to the final one are precisely those that contain the letter \(2\) at least once.
The language consisting of such words can be represented by the expression \(\{1,2\}^*2\{1,2\}^*\).
 \end{example}

\begin{figure}[h]
	\centering
    \begin{tikzpicture}[node distance = 3cm, on grid, auto]
      \node[state, initial, initial text = {}] (q0) {\(s_0\)};
    \node[state, accepting, right of=q0] (q1) {\(s_1\)};
        \draw[->] (q0) edge[above] node{2} (q1)
        (q1) edge[loop above] node{1,2} (q1)
        (q0) edge[loop above] node{1,2} (q0);
    \end{tikzpicture} 
	\caption{\label{fig:dfaK2} dfa recognizing \(\{1,2\}^*2\{1,2\}^*\).}
\end{figure} 
 
A step higher on the hierarchy, we find context-free languages, which are recognized by a type of automata known as nondeterministic pushdown automata or ndpa. These machines have a higher capability than dfa's because they have a memory (in the form of a stack).
 On top of these languages, we have recursively enumerable languages. Turing machines accept the languages in this family. And if instead of regular Turing machines, we look at the languages accepted by nondeterministic Turing machines with the tape limited to the input we get to context-sensitive languages. These automata are known as linear bounded automata or lba and play an essential role in this article. It is worth noting that as long as the tape is limited to a number of cells that depends linearly on the input, the kind of language defined is the same.

%%%
\subsection{Kunz words}\label{subsec:Kunz-words}
We shall translate the representation referred to in Section~\ref{subsec:numericalsgps} of numerical semigroups by tuples of integers into a representation by words over the alphabet \(\mathbb{N}_{\ge 1}\) of positive integers. 

This will lead us to associate a language to the class \(\mathcal{S}_q\) of numerical semigroups of depth \(q\) and then study it in the light of the classical theory of formal languages, namely where it fits in the Chomsky Hierarchy of languages.

The Kunz coordinate tuple \(\mathcal{K}(S)=\left(k_1, k_2, \ldots, k_{m-1}\right)\) associated with a numerical semigroup of depth \(q\) can be seen as a word \(\mathcal{W}(S) = k_1 k_2 \ldots k_{m-1}\) over the alphabet \(\{1,\ldots,q\}\).
This association already appears in the literature (see, for instance, \cite{Zhu2023CT-Sub}).
Notice that \(\operatorname{m}(S)=|\mathcal{W}(S)|+1\). 
The word \(\mathcal{W}(S)\) satisfies a number of relations imposed by the inequalities the Kunz coordinates tuple of \(S\) satisfies (see the inequalities~(\ref {eq:kunz-inequalities}) on page~\pageref{eq:kunz-inequalities}). This leads us to the following formal definition:

\begin{definition}\label{def:kunz-word}
  A word \(u=u_1\cdots u_{\ell}\in    (\mathbb{N}_{\ge1})^*\) is said a \emph{Kunz word} if
	\begin{align*}
		u_i+u_j & \ge u_{i+j}, \\
		u_i+u_j+1 & \ge u_{i+j-(\ell+1)},
	\end{align*}
	for all indices for which the above inequalities are defined.
	
	In this context, we say that the empty word is the \emph{\(0\)-Kunz word} and we say that a Kunz word \(u = u_1\cdots u_{\ell}\) is a \emph{\(q\)-Kunz word} if \(\max\{u_1,\ldots,u_{\ell}\}=q\).
\end{definition}

The inequalities in the above definition will be referred to as \emph{Kunz conditions}.

From the above discussion, we have that a numerical semigroup of depth \(q\) and multiplicity \(m\) corresponds, through \(\mathcal{W}\), to exactly one \(q\)-Kunz word of length \(m-1\). Notice that the \(0\)-Kunz word corresponds to \(\mathbb{N}\).  

It follows that the map \(\mathcal{W}\) is a bijection between the class of all numerical semigroups (respectively, numerical semigroups of depth \(q\)) and the class of all Kunz words (respectively, \(q\)-Kunz words). Thus \(\mathcal{W}\) provides a representation of numerical semigroups by words.

\begin{example}\label{ex:kunz-words}
  The set \(\{1,2\}^*\) consists of Kunz words, since a word consisting entirely of \(1\)s and \(2\)s satisfies the Kunz conditions, as the left members of the inequalities are at least \(2\).
\end{example}

\begin{lemma}\label{lemma:kunz-and-non-kunz-words}
	Let \(n,m\) and \(q\) be positive integers, with \(q\geq 3\). Then
	\begin{enumerate}[label = \textnormal{(\roman*)}]
	 \item the word 
  \( u = 1^n2^n3^n\dots (q-2)^n(q-1)^nq\)
  is \(q\)-Kunz;\label{itm:q-kunz}
	 \item the word 
  \( v = 1^{n+m}2^n3^n\dots (q-2)^n(q-1)^nq\) 
  is not Kunz.\label{itm:not-kunz}
	\end{enumerate}
\end{lemma}
\begin{proof}
	Denote by
  \( u_{\kappa} \,(\text{respectively, } v_{\kappa})\)
 the \(\kappa\)th letter of the word
  \( u \,(\text{respectively, } v)\).  
	
 We now prove part~\ref{itm:q-kunz}.
  It is clear that \(\lceil x \rceil+\lceil y \rceil\ge\lceil x +y \rceil\), for any \(x,y\in\mathbb{R}\).\\
 Note that \(|u| = (q-1)n+1\). Observe also that \(u_{\kappa}=\left\lceil\frac{\kappa}{n}\right\rceil\) for all indices~\(\kappa\). It follows that if \(i\) and \(j\) are indices such that \(1\le i+j \le (q-1)n+1\), then
\[u_i+u_j=\left\lceil\frac{i}{n}\right\rceil+\left\lceil\frac{j}{n}\right\rceil\ge \left\lceil\frac{i+j}{n}\right\rceil = u_{i+j},\]	
thus \(u\) satisfies the first Kunz condition.
		
		If \(i\) and \(j\) are such that \(i+j > (q-1)n+1\), then 
		\[\left\lceil\frac{i}{n}\right\rceil+\left\lceil\frac{j}{n}\right\rceil\ge\left\lceil\frac{i+j}{n}\right\rceil \ge \left\lceil\frac{(q-1)n+1}{n}\right\rceil = q-1 + \left\lceil\frac{1}{n}\right\rceil = q.\] 
		It follows that \(u\) also satisfies the second Kunz condition.
			
	Next we prove part \ref{itm:not-kunz}.
        If we take \(i= n + 1\) and \(j = n+m\), we have that \(v_i +v_j = 2\), but \(v_{i+j} =v_{2n + m +1} = 3\), so \(v\) does not meet the first condition to be a Kunz word.
\end{proof}

\begin{definition}
	Let \(q\ge 0\) be an integer.
	We define the \emph{\(q\)-Kunz language} (sometimes referred to as \emph{Kunz language of depth~\(q\)}) as the set of \(q\)-Kunz words over the alphabet \(\{1,\ldots,q\}\). We denote this language by \(K_q\).
	Using the language of sets, we can write:
	\begin{multline*}
   	 	K_q=\left\{u=u_1\cdots u_r\in \{1,\ldots,q\}^*\mid q\in\{u_1,\ldots,u_r\}, u_i+u_j\ge u_{i+j} \text{ and } u_i+u_j+1\ge u_{i+j-(|u|+1)},\right.\\
    	\left.\text{ for all indices for which these inequalities are defined}\right\}.
	\end{multline*}

	We say that a language is a \emph{Kunz language} if it is a \(q\)-Kunz language for some nonnegative integer~\(q\).
\end{definition}

The map \(\mathcal{W}\) gives a bijection between the set of numerical semigroups of depth \(q\) and the \(q\)-Kunz language.
\medskip

Whether Kunz languages have some interesting properties is a question we will try to answer positively in this article.
  As very little can be done with languages for which we cannot find efficient membership algorithms, before proceeding, we exhibit (although without efficiency concerns) a membership algorithm for Kunz Languages.   

\begin{fact}
  Any Kunz language has a solvable membership problem.
\end{fact}
\begin{proof}
  Algorithm~\ref{alg:membership-problem-for-Kunz-languages} is a naive algorithm that tests whether a word over \(\mathbb{N}_{\ge1}\) is a \(q\)-Kunz word.
\end{proof}

\begin{algorithm}[ht]  \caption{Algorithm to test whether a word is \(q\)-Kunz.\label{alg:membership-problem-for-Kunz-languages}}
\KwData{\(q\in\mathbb{N}, v=v_1\cdots v_{\ell}\in (\mathbb{N}_{\ge1})^*\)}
\KwResult{\textsf{true} if \(v\) is a \(q\)- Kunz word and \textsf{false} otherwise}

\If{\(|v|=0 \text{ and } q=0\)}{\KwRet{true};}
\If{\(\max\{v_i\mid i\le |v|\}\ne q \)}{\KwRet{false};}

\(pairs = \{(i,j)\in\{1,\ldots,|v|\}^2\mid i\le j\}\)

\For{\((i,j)\in pairs\)}{
    \If{\(i+j\le |v|\)}{
        \If{\(v_i+v_j < v_{i+j}\)}{
            \KwRet{false}\;
        }
    }
    \Else{
        \If{\(v_i+v_j +1< v_{i+j - (|v|+1)}\)}{
            \KwRet{false}\;
        }
    }
}
\KwRet{true}
\end{algorithm}

\section{Kunz languages}\label{sec:Kunz-languages}
Along this section, we aim to prove that \(q\)-Kunz languages are regular if \(q<3\) (Proposition~\ref{prop:k0k1k2regular}). and they are not when \(q\ge 3\) (Proposition~\ref{prop:3_not_regular}).
  In addition, we prove that for \(q\ge 5\) the language \(K_q\) is not context-free (Proposition~\ref{prop:K5_not_context-free}).
  Moreover, we will observe that for \(q\ge 3\) the language \(K_q\) is context sensitive (Proposition~\ref{prop:K3-context-sensitive}).
  Also, we leave two open problems: Problem~\ref{problem:is_K3_context-free}, asking whether \(K_3\) or \(K_4\) are context-free, and Problem~\ref{prob:give-classes-of-languages}, asking for the consideration of classes of languages between context-free and context-sensitive and their correspondence with classes of semigroups of distinct depths.
  
    Our results suggest that problems in numerical semigroups of depth up to \(2\) should be more accessible than the same problems for numerical semigroups of higher depth.
This seems to be the case, as we illustrate next by means of examples related to counting numerical semigroups and to Wilf's conjecture.

For example, Zhao~\cite{Zhao2010SF-Constructing} proved that the number of numerical semigroups \(S\) with genus \(g\) and depth up to \(2\) is the Fibonacci number \(F_{g+1}\) (where the Fibonacci sequence is considered starting in \(0\)). That the number of numerical semigroups \(S\) with genus \(g\) and depth up to \(3\) behaves asymptotically like the Fibonacci sequence was later proved by Zhai~\cite{Zhai2012SF-Fibonacci}, and the proof is quite more involving. See an article by Kaplan~\cite{Kaplan2017AMM-Counting} for a survey on the theme, with particular emphasis on the above-mentioned results. 
	As another example, we refer to the fact that it is known that Wilf's conjecture holds for numerical semigroups of depth up to \(2\) since 2012 (the result was obtained by Kaplan~\cite{Kaplan2012JPAA-Counting}), but only six years later has been obtained the analogous result for numerical semigroups of depth up to \(3\), a result due to Eliahou~\cite{Eliahou2018JEMS-Wilfs}.

It is reasonable to expect a growth in the difficulty of problems in classes of numerical semigroups as the depth grows. Still, our results do not provide a distinction in complexity between the languages \(K_q\) as long as \(q\ge 3\). To support our expectation, we refer to the fact that it is known that Wilf's conjecture holds for numerical semigroups of depth \(3\), but the proof is considerably more involved than the one for numerical semigroups of depth \(2\). It is still unknown whether Wilf's conjecture holds for numerical semigroups of depth \(4\).

As our results do not distinguish between certain languages (thus between the corresponding classes of numerical semigroups), it is natural to consider intermediate classes of languages between context-free and context-sensitive and their correspondence with classes of semigroups of distinct depths as a future research topic.

\begin{problem}\label{prob:give-classes-of-languages}
    Consider intermediate classes of languages between context-free and context-sensitive and their correspondence with classes of semigroups of distinct depths.
\end{problem}

%%%%%%%%%%%%%%%
\begin{fact}
  \begin{enumerate}
  \item  \(K_0\) consists of the empty word;
    \item \(K_1=\{1\}^+\);
    \item \(K_2=\{u2v\mid u,v\in \{1,2\}^*\}=\{1,2\}^*2\{1,2\}^*\)
\end{enumerate}
\end{fact}

    The reader familiar with formal languages readily observes that the preceding fact exhibits regular expressions for \(K_0\), \(K_1\) and \(K_2\), thus showing that these languages are regular.
    We state this fact next, and give an alternative proof, to keep the paper self-contained.
%  }
\begin{proposition}\label{prop:k0k1k2regular}
    The Kunz languages \(K_0\), \(K_1\) and \(K_2\) are regular.
\end{proposition} 
\begin{proof}
  The right linear grammar with the production \(S\to\varepsilon\), where \(\varepsilon\) is the empty word, is \(K_0\). Thus, \(K_0\) is regular. 
  The language \(K_1\) is regular, since it is the language of the right linear grammar with the production \(S\to 1\, S\).
    That the language \(K_2\) is regular follows from Example~\ref{ex:dfaK2}, which shows that it is a language accepted by a dfa.
\end{proof}

The language \( K_3 \) is more complex. It does not belong to the class of languages at the bottom of the Chomsky hierarchy, as stated in the following proposition. 
In the proof we make use of Lemma~\ref{lemma:kunz-and-non-kunz-words}, which provides examples of Kunz words and examples of non-Kunz words.
\begin{proposition} \label{prop:3_not_regular}
	Let \(q\ge 3\) be an integer. Then \(K_q\), the Kunz language of depth \(q\) is not regular. 
\end{proposition}

\begin{proof}
    For the sake of contradiction, let us suppose that \(K_q\) is a regular language. Then, there has to be a dfa \(M\) that recognizes it. Let \(n\) be the number of states of such dfa.\\
    The word \(w = 1^n2^n3^n\dots (q-2)^n(q-1)^nq\) is a Kunz word of depth \(q\) so it has to be accepted by \(M\) but since \(M\) has only \(n\) states, while reading the first \(n\) \(1\)s of \(w\), \(M\) has to come at least twice to the same state \(s_1\) due to the pigeonhole principle. In other words, while reading the \(1\)s in \(w\), a cycle \(s_1\rightarrow s_2 \rightarrow \cdots \rightarrow s_i \rightarrow s_1\) of length \(i\) is walked. Consequently, any word adding any multiple of \(i\) \(1\)s will also be accepted. Then for \(i>0\) the word \(1^{n+i}2^n3^n\dots (q-2)^n(q-1)^nq\) is accepted by \(M\), which is a contradiction since \(1^{n+i}2^n3^n\dots (q-2)^n(q-1)^nq\) is not a Kunz word.
\end{proof}
  The pigeonhole argument used in the above proof can be used to prove the classical Pumping lemma for regular languages, which can be paraphrased as follows:
  every sufficiently long string in an infinite regular language \(L\) can be broken into three parts in such a way that an arbitrary number of repetitions of the middle part yields another string in \(L\). We say that the middle string is ``pumped'', hence the term pumping lemma for this result.

The above-mentioned Pumping lemma for regular languages has many versions and generalizations (including one for context-free languages). We successfully used one of them, by Bader and Moura~\cite{bader-moura}, to prove that Kunz languages of depth greater than \(4\) are not context-free (Proposition~\ref{prop:K5_not_context-free}).
On the other hand, we did not succeed in using any of the available generalizations to prove that \(K_3\) is not context-free. We leave the following open problem:
  \begin{problem}\label{problem:is_K3_context-free}
Is any, or both, of the languages \(K_3\) or \(K_4\) context-free?
    \end{problem}
Before stating the above-mentioned Bader-Moura's result we need to introduce some notation.
\\
Given a language \(L\) and a word \(w\in L\) we may label some positions (letters) in \(w\) as \textit{distinguished} and some others as \textit{excluded}. A particular position can be both distinguished and excluded or maybe neither. When labelling positions of \(w\), we will denote \(d(u)\) (resp. \(e(u)\)) as the number of distinguished (resp. excluded) positions in a substring \(u\) of \(w\).

\begin{theorem}\cite{bader-moura} \label{th:Bader-Moura}
For any context-free language \(L\), there exists \(p\in \mathbb{N}\) such that for every \(w\in L\), if \(d(w)\) positions are labeled as distinguished and \(e(w)\) are labeled as excluded with \(d(w)>p^{e(w) + 1}\), then there exist \(u,v,x,y,z\) such that \(w = uvxyz\) and:\begin{enumerate}
    \item \(d(vy)\geq 1\) and \(e(vy) = 0\).
    \item \(d(vxy) \leq p^{e(vxy)+1}\).
    \item \(uv^\kappa xy^\kappa z \in L\) for every \(\kappa \geq 0\).
\end{enumerate}
\end{theorem}

\begin{proposition}\label{prop:K5_not_context-free}
For \(q\geq 5\) the Kunz language of depth \(q\) is not context-free.
\end{proposition}

\begin{proof}
We will argue by contradiction. Suppose that \(K_q\), the Kunz language of depth \(q\), is context-free. Then by Theorem~\ref{th:Bader-Moura}, there exists a \(p\in \mathbb{N} \) satisfying the conditions of the theorem. Let \(n = p^{q} +1\). By Lemma~\ref{lemma:kunz-and-non-kunz-words}, the word \(\; w = 1^n2^n3^n\dots (q-2)^n(q-1)^nq \;\) belongs to \(K_q\).\\
  Now, we label all the positions with the \(1\)s as distinguished and for every \(2\leq a \leq q\) we label the first appearance of \(a\) in \(w\) as excluded.
Next we can see \(w\) represented with the distinguished positions in green and the excluded ones in red.

\[w= \underbrace{\textcolor{green}{11\cdots1}}_{n} \textcolor{red}{2} \underbrace{22\cdots 2}_{n-1}\textcolor{red}{3}\underbrace{33\cdots 3}_{n-1}\textcolor{red}{4}4\cdots(q-2)\textcolor{red}{(q-1)}\underbrace{(q-1)(q-1)\cdots (q-1)}_{n-1}\textcolor{red}{q} \]
We have \(d(w)= n = p^{q} +1\) and \(e(w) = q-1\), hence \(d(w) = p^{q} +1 > p^{q} = p^{e(w) + 1}\). As a consequence, the three statements of Theorem~\ref{th:Bader-Moura} must hold.
\\
So we will have \(w = uvxyz\) where \(v y\) contains at least one \(1\). Since by the first statement \(e(vy) = 0\) we have that neither \(v\) nor \(y\) contain two different letters, and at least one of them has to be completely made of \(1\)s. So either \(vxy = 1^k\) for some \(k > 0\) or \(v= 1^k\) and \(y = a^s\), for some \(k > 0\), \(2\le a < q\) and \(s\ge 0\).\\
In the case \(vxy = 1^k\), by the third statement of Theorem~\ref{th:Bader-Moura} we could pump \(v\) and \(y\) and get \({1^{n+t}2^n3^n\dots (q-2)^n(q-1)^nq}\) which is not Kunz, by the proof of Lemma~\ref{lemma:kunz-and-non-kunz-words}.\\
In the other case, \(v= 1^k\) and \(y = a^s\), we have two options. The first option is that \(y\) consists entirely of \(2\)s. If this is so, we could pump \(y\) until getting two indices \(i,j\) such that \(w_i,w_j \leq 2 \) but \(w_{i+j} \geq 5\).\\
If \(y\) is of the form \( a^{s}\) for some \(a\geq 3\) and \(s\ge 1\), we could pump \(v\) sufficiently until getting two indices \(i,j\) such that \(w_i=w_j=1\) but \(w_{i+j}\ge 3\).
\end{proof}

\begin{remark}
    Note that in Proposition \ref{prop:K5_not_context-free} in order for the argument to work, the condition \(q\geq 5\) is necessary. This is because in the case \(v = 1^m\) and \(y=2^m\) the word \(uv^\kappa xy^\kappa z\) would still be Kunz. That is, pumping a lot of \(2\)s gets to a non Kunz word in the case of \(K_5\) because \(2+2 < 5\) but it would not make any difference in the cases of \(K_3\) and \(K_4\) since \(2+2 \geq 4\), and the condition \(w_i+w_j\geq w_{i+j}\) would still hold.
\end{remark}

Next, we will see that \(K_3\) is context-sensitive by giving an lba capable of accepting it. Since the machine required for this task is pretty complex, we will not give a formal definition; instead, a high-level description of the automaton will be provided. In this high-level description, we will specify some steps that an lba can do, hence describing its behaviour but not its states and moves. Then, similarly, we will give an idea of how to do the same thing for any Kunz language \(K_n\) for \(n\in \mathbb{N}\) and, as a consequence, prove that all Kunz languages are context-sensitive. Full details will appear in the second author's master thesis~\cite{UsoiCubertorer2023MThesis-Kunz}.

\begin{proposition}\label{prop:K3-context-sensitive}
The Kunz language of depth \(q=3\) is context-sensitive.
\end{proposition}
\noindent  \textbf{Idea of an lba for \(\mathbf{K_3}\).}
We give a high-level description of an lba with each cell of the tape divided into five parts, called tracks, that accepts \(K_3\). The first track contains the input, and the second and third tracks will be used to store in unary notation indices of \(1\)s in the word \(w\). The fourth track is used to eventually compute and store the result of the sum of the indices of the previous tracks. Finally, the last one just stores the length of the word \(w\) in unary notation. The amount of tape used will be inferior to \(6\ell\) where \(\ell = |w|\) is the length of the input word. So at the very beginning, the first track will be \(\square w_1w_2\cdots w_\ell \square \cdots \square\) and the other tracks will be completely filled with blanks. The first blank right before \(w_1\) will remain untouched since it marks the left end of the tape. We could substitute it with a symbol distinguishing the left end of the tape, such as the symbol \([\) used by some authors to define an lba. The same happens with the rightmost \(\square\); it serves as a \(]\); we just used blanks to simplify the notation.
\\
Notice that in order for a word \(w = w_1w_2\cdots w_n\in \{1,2,3\}^*\) to be Kunz it is necessary and sufficient that \(w_i + w_j \geq w_{i+j}\). To contradict this fact, the only possible way is that \(w_i = w_j = 1\) and \(w_{i+j} = 3\). Then in order for \(w\) to be in \(K_3\) there is the additional condition that \(3\in w\).\\
The algorithm that the machine will execute is the following: first, check that the word contains at least one \(3\). Then from the beginning of the word going left to right, find the first \(1\), cross it, write its position (\(i\)), check if \(w\) has a \(3\) in position (\(i+i\)), go back to the \(1\) we just crossed, and check for every other \(1\) in position \(j\) after the \(x\) if the word contains a \(3\) in position \(i+j\). Once all other \(1\)s are verified for this first \(1\), do the same thing with the others.\\

The machine is described using steps. These steps are numbered and followed in order unless a step requires to move to another step that is not the one after. 
It is an easy exercise to verify that these are all within the capabilities of a lba.

\subsection*{Steps}
\begin{enumerate}
\item 
  Go all the way up until the end of the word to check if there is at least one \(3\) and go back to the beginning of the tape. The word is rejected if no \(3\) is found.
While going forward, at every move, put a \(1\) in the cell on the fifth track, in order to store the length of the word in that track.
    \item Go right until the first \(1\) is found and at every step right, write a \(1\) in the same cell but in the second and third tracks. If there are not any \(1\)s, accept the word. If a \(1\) is found, change it by an \(x\) (cross the \(1\)) and still write \(1\)s in the second and third tracks. \emph{At this point, we will have the one we are about to check marked with an \(x\) and its index \(i\) written in the second and third tracks.}
    \item  On the fourth track, compute and record \(i+i\) (in unary). Then go to the position with that index (\(2i\)) and check if it is a \(3\), if so, reject the word; otherwise (or in the case \((i+i)\) exceeds the length of the word) go back left to the first \(x\). 
    \item Go right until the next \(1\) is found. While doing it, at each step write on the same cell but in track \(3\) a \(1\). If there is not any other \(1\), accept the word. In case there is another \(1\), change it to a \(y\) and write a \(1\) on track 3. At this point, we have marked with a \(y\) the \(1\) we are about to check along with the \(1\) we already marked with an \(x\). Since the second track remained untouched, it still contains \(i\), the index of \(x\), while on the third track, we have \(j\), the index of \(y\) in unary notation.
    \item On the fourth track, compute and store the sum of the two indices (\(i + j\)).
    \item Go to the position \(i+j\) on the first track and check its value. If it is a \(3\), reject the word; otherwise (or in the case \((i+j)\) exceeds the length of the tape), go back to the first \(y\) found to the left.
    \item 
    Look for the next \(1\); if it is found, change it by a \(y\) and go to \textbf{step 5}. If there is not any other \(1\) on the tape, change all \(y\) on the tape back to \(1\), position the head of the tape on the rightmost \(x\) and go to \textbf{step 3}. 
    While going left to find the rightmost \(x\), it is important to clear every 1 in the third track the head goes through, in order to have the index of that \(x\) on both the second and third tracks.
\end{enumerate}

What we have done previously is essentially describe the technicalities needed to implement Algorithm~\ref{alg:membership-problem-for-Kunz-languages} in an lba. Another way to visualize and explain the machine needed to describe \(K_3\) is through a transition graph like the one shown in Figure~\ref{fig:lbaK3}.

\begin{figure}[ht]
  \centering
  \includegraphics[width=0.80\textwidth]{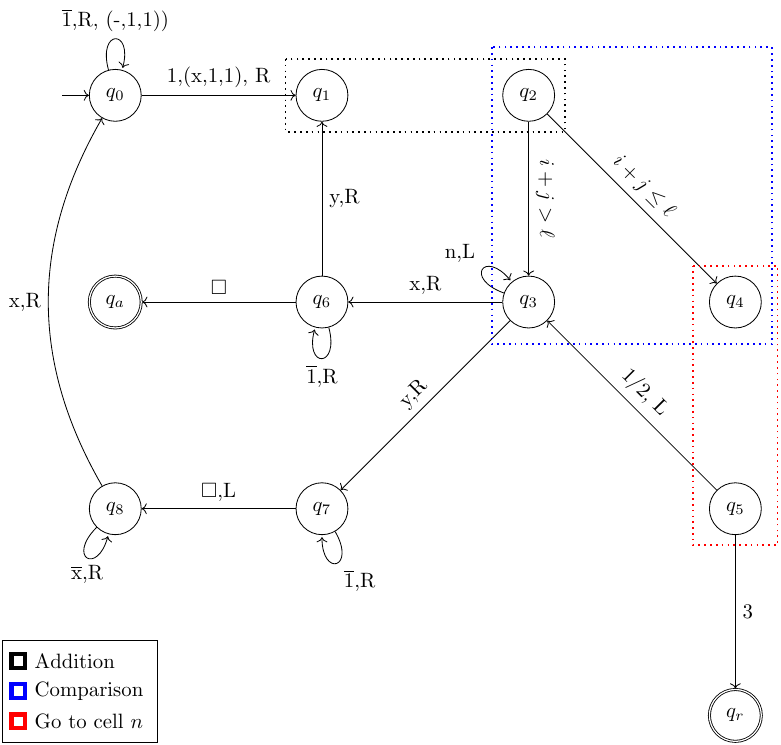}

	\caption{\label{fig:lbaK3} lba that accepts \(K_3\).}
\end{figure} 
In Figure \ref{fig:lbaK3}, the black dotted rectangle represents the subroutine of adding two numbers in unary notation. This corresponds to the input in the second and third tracks. So we shall imagine that when the machine enters state \(q_1\), it computes the addition and returns the output in track four, halting in the state \(q_2\). At this point, we call a second subroutine, the comparison of numbers. This subroutine is marked with the blue square and contains three states: \(q_2,q_3\) and \(q_4\). After computing the addition of indices, this comparing subroutine checks whether the result is greater than \(\ell\) (the length of the word \(w\)) and halts in either \(q_3\) or \(q_4\). Finally, in red, we represent the subroutine of going to a specific position in the word indexed by a natural number in unary notation. This corresponds to the addition of two indices \(i + j\) in which there are \(1\)s, in the case that the sum does not exceed \(\ell\). 
\begin{remark}
    In Figure \ref{fig:lbaK3}, \((-,1,1)\) means leave the first track untouched and write a \(1\) in the second and third tracks, and similarly, \((x,1,1)\) stands for the same, but writing an \(x\) on the first track. Then, a bar over a symbol such as \(\overline{x}\) means anything but that symbol, in this particular case, anything except for an \(x\). Then, in the case of the state \(q_3\), while moving left, it is not written on the graph, but the second and third tracks should be cleared as specified in step 7. 
\end{remark}

The version for \(K_3\) makes it much more simple since we only need to check one of the Kunz conditions and moreover just observe positions with \(1\)s and \(3\)s. 

Dealing with \(K_n\) for \(n>3\) requires more effort but, with the same idea, one can implement a full version of Algorithm~\ref{alg:membership-problem-for-Kunz-languages} on an lba. 

The most important fact to take into account is that given \(n\geq 3\), the computations needed to run the algorithm to decide whether or not a word \(w\) is in \(K_n\) consist of simple addition, subtraction and comparison of natural numbers and moving a specific number of cells within the tape. All these can easily be implemented on a Turing machine, and it is not hard to notice that the amount of tape needed depends linearly on the length of the input. As we are working on an alphabet \(1,2,\dots,n\), the values of symbols in the word will not exceed \(n\), and the values of indices will not exceed \(\ell\) the length of the word, so using \(2n\ell\) cells is enough.

\subsection*{Acknowledgements}
The authors would like to thank Pedro García-Sánchez for his very useful comments. 
The authors also thank the reviewer for her/his comments and suggestions.
%%%%%%%%%%% References %%%%%%%%
% in order to use bibtex comment the following line and uncomment the others
%% uncomment to use biblatex  
\printbibliography
%%
%%uncomment to use bibtex
%\bibliographystyle{plainurl}
%\bibliography{refs.bib}
%%%%%%%%%
%\input{refs}
%%%%%%%%%%%%%%%%%%%%%%%%%%%%%%%%%%%%%%%%%%%%%%%%%%%%%%%%%%%%%%%%%%%%%
\end{document}